\documentclass[
  prl,
  twocolumn,
  superscriptaddress,
  citeautoscript,
  showpacs,
  amsart,
  longbibliography
]{revtex4-1}

\usepackage{amsmath,bm,amsfonts}
\usepackage{physics,mathtools}
\usepackage{graphicx}
\usepackage[normalem]{ulem}
\usepackage[colorlinks=true,linkcolor=blue,citecolor=blue]{hyperref}

\usepackage[dvipsnames]{xcolor}

\usepackage[english]{babel}
\usepackage{amsthm} %for proof

\begin{document}

\title{\textbf{Many-Body Resonances in the Avalanche Instability of Many-Body Localization}}

\author{Hyunsoo Ha}
\affiliation{Department of Physics, Princeton University, Princeton, New Jersey 08544, USA}

\author{Alan Morningstar}
\affiliation{Department of Physics, Princeton University, Princeton, New Jersey 08544, USA}
\affiliation{Department of Physics, Stanford University, Stanford, California 94305, USA}

\author{David A. Huse}
\affiliation{Department of Physics, Princeton University, Princeton, New Jersey 08544, USA}

\date{\today}

\begin{abstract}
Many-body localized (MBL) systems fail to reach thermal equilibrium under their own dynamics, even though they are interacting, nonintegrable, and in an extensively excited state.  One instability toward thermalization of MBL systems is the so-called ``avalanche,'' where a locally thermalizing rare region is able to spread thermalization through the full system. 
The spreading of the avalanche may be modeled and numerically studied in finite one-dimensional MBL systems by weakly coupling an infinite-temperature bath to one end of the system.  We find that the avalanche spreads primarily via strong many-body resonances between rare near-resonant eigenstates of the closed system.  Thus we find and explore a detailed connection between many-body resonances and avalanches in MBL systems. 
\end{abstract}

\maketitle
%%%%%%%%%%%%%%%%%%%%%%%%%%%%%%%%%%%%%%%%%%%%%%%%%%%%%%%%%%%%%%%%%

\emph{\textbf{Introduction---}}
Many-body localized (MBL) systems are a class of isolated many-body quantum systems that fail to thermalize due to their own unitary dynamics, even though they are interacting, nonintegrable and extensively excited~\cite{anderson_mbl,basko_altshuler_mbl,oganesyan_Huse_mbl,nandkishore_huse_mbl,deroeck_imbrie_mbl_review,alet_laflorencie_mbl_review,abanin_maksym_mbl_colloquium}.
This happens for one-dimensional systems with short-range interactions in the presence of strong enough quenched randomness, which yields a thermal-to-MBL phase transition of the dynamics.  In the MBL phase, there are an extensive number of emergent localized conserved operators \cite{serbyn_abanin_LIOM,huse_Oganesyan_lbit,chandran_abanin_liom_construction,Imbrie-Scardicchio2017_liom_review}.

One instability of the MBL phase that is believed to play a central role in the asymptotic, long-time, infinite-system MBL phase transition is the ``avalanche'' \cite{deroeck_huveneers_avalanche,thiery_deroeck_avalanche,thiery_deroeck_avalanche2,morningstar_huse_imbrie_RGcritical}.
Rare locally thermalizing regions necessarily exist due to the randomness. 
Starting from such a rare region, this ``thermal bubble" spreads through the adjacent typical MBL regions until the relaxation rate of the adjacent spins becomes smaller than the many-body level spacing of the thermal bubble, in which case the avalanche halts.  If the strength of the randomness is insufficient, the relaxation rate remains larger than the level spacing and the avalanche does not stop: the full system then slowly thermalizes and is no longer in the MBL phase (although it is in a prethermal-MBL regime~\cite{alan_huse_avalanche,long_anushya_jacobi_resonance}).

\begin{figure}
    \centering
	\includegraphics[width=0.5\textwidth]{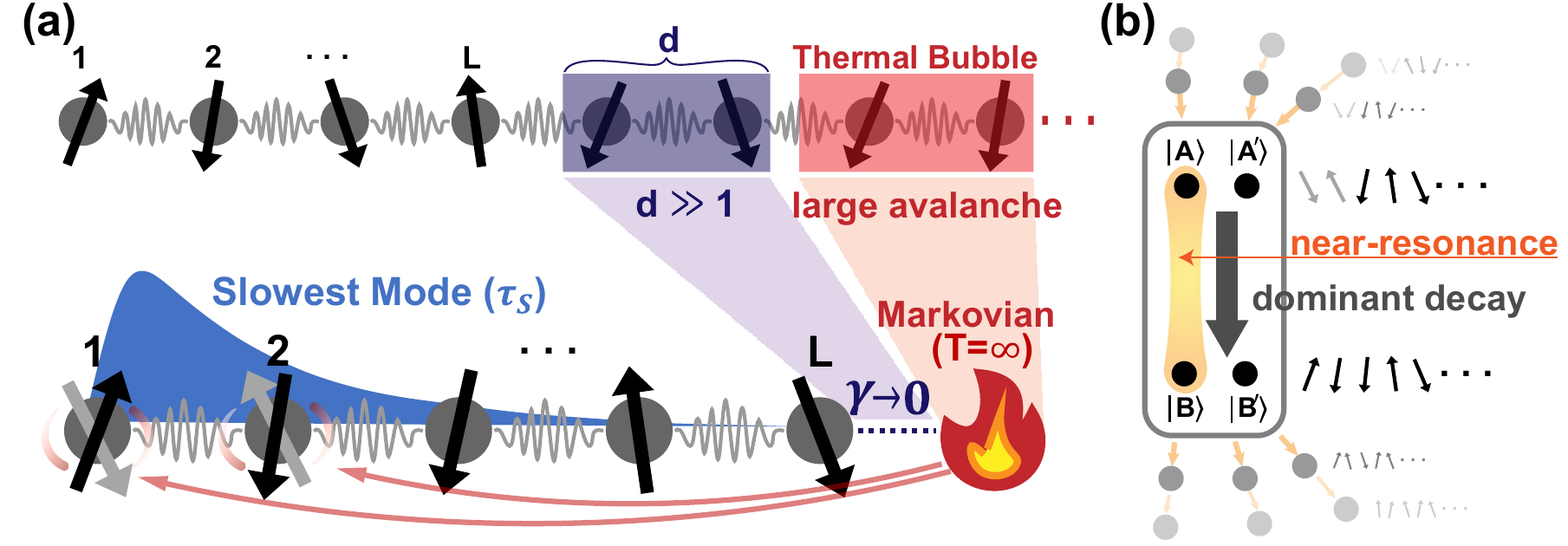}
	\caption{(a) The avalanche model. We model the large avalanche (seeded by a bare thermal rare region) spreading, but still far away ($d\gg 1$) by connecting the Markovian infinite-temperature bath in the weak-coupling limit with the one-dimensional MBL system of length $L$ spins. Specifically, we analyze the decay of the slowest mode ($\hat{\tau}_{S}$), which is localized near the end of the system farthest from the bath; $\hat{\tau}_{S}$ is a ``localized integral of motion''  in the MBL phase. (b) Schematic decay of $\hat{\tau}_{S}$.  A large fraction of the probability current in the decay of $\hat{\tau}_{S}$ passes through four eigenstates associated with a rare near-resonance.}
 \label{fig:cartoon}
\end{figure}

The avalanche has been numerically simulated in small systems \cite{luitz_deroeck_avalanche_modelling,goihl_Krumnow_avalanche_modelling,crowley_chandran_avalanche_modelling,alan_huse_avalanche,sels_avalanche,Tu-Sarma2022_avalanche} and experimentally probed~\cite{leonard_greiner_avalanche_exp}.
Recent work shows that the instability of MBL to avalanches occurs at much stronger randomness than had been previously thought \cite{alan_huse_avalanche,sels_avalanche}.  This leaves a large intermediate prethermal-MBL regime in the phase diagram between the onset of MBL-like behavior in small samples (or correspondingly short times) and the asymptotic MBL phase transition.  Clear numerical evidence has been obtained for many-body resonances being an important part of the physics in the near-thermal part of this regime~\cite{alan_huse_avalanche,gopalakrishnan_huse_resonance,crowley_chandran_resonance,khemani_huse_criticalmbl,khemani_huse_quasimbl,geraedts_regnault_mbl_entanglementspectrum,long_anushya_jacobi_resonance,villalonga_clark_mblresonance,garratt_chalker_mblresonance}, while no such evidence for the expected thermalizing rare regions has been found yet.  In the part of this intermediate prethermal regime that is farther from the thermal regime, it remains unclear what the primary mechanism that leads to thermalization is for samples larger than those that can be diagonalized.

In this work, we explore how an avalanche spreads through typical MBL regions for systems that are near the avalanche instability. In particular, we uncover explicit connections between many-body resonances and avalanches. We do not simulate the rare region that initiates the avalanche.  Instead, we assume a large avalanche is spreading, and that we may model that as an infinite-temperature bath (see Fig.~\ref{fig:cartoon}) weakly coupled to one end of our MBL spin chain~\cite{alan_huse_avalanche,sels_avalanche} (similar settings were also considered for studying transport~\cite{Znidaric_mbltransport1,Znidaric_mbltransport2,Znidaric_mbltransport3,fischer_altman_mblbath}). We find that particular many-body near-resonances of the closed system play a key role in facilitating the avalanche.  These near-resonances are the dominant process by which the bath at one end of the chain thermalizes the other end of the chain and thus propagates the avalanche.

\emph{\textbf{Model---}}
Our model consists of a chain of $L$ spin-1/2 degrees of freedom.  The dynamics of the closed system is given by the random-circuit Floquet MBL model of Ref.~\cite{alan_huse_avalanche}, which has unitary Floquet operator $\hat{U}_F$.  The disorder strength in this model is given by the parameter $\alpha$, with localization occurring at large $\alpha$ [see the Supplemental Material (SM)~\cite{supp_mat} for model details].  To investigate avalanche spreading, we weakly connect an infinite-temperature Markovian bath to spin $L$ at the right end of the system~\cite{alan_huse_avalanche,sels_avalanche}.  The quantum state of this open system is the density matrix $\hat{\rho}(t)$.

In our open-system Floquet model, the bath is represented by the superoperator $S_\mathrm{bath}$ that acts once each time period:
\begin{align}
\label{eq:S_bath}
S_{\mathrm{bath}}[\hat{\rho}] = \frac{\hat{\rho}}{1+3\gamma} + \frac{\gamma}{1+3\gamma}\sum_{j=1}^3 \hat{E}_j \hat{\rho} \hat{E}_j^\dagger~,
\end{align}
where $(\hat{E}_1,\hat{E}_2,\hat{E}_3) = (\hat{X}_L, \hat{Y}_L, \hat{Z}_L)$ are the jump operators acting on the last spin at site $L$. We will take the weak-coupling limit $\gamma \rightarrow 0$, representing the exponentially weak coupling of spin $L$ to the distant thermal rare region. The open-system Floquet superoperator $S_{\mathrm{period}}$ that takes our system through one time period is $S_{\mathrm{period}}[\hat{\rho}(t)]=S_{\mathrm{bath}}[\hat{U}_F \hat{\rho}(t) \hat{U}_F^\dagger]$.
The time evolution of the system's state is given by \begin{align}
    \hat{\rho}(t) = \hat{\mathbb{I}}/2^L + p_1 e^{-r_1t}\hat{\tau}_{1} + \sum_{k\geq2}p_k e^{-r_k t}\hat{\tau}_{k},
\end{align}
where $e^{-r_k}$ is the $k$th largest eigenvalue of $S_\mathrm{period}$ with eigenoperator $\hat{\tau}_k$. Note that the largest (0th) eigenvalue is 1 and is nondegenerate when $\gamma>0$, with eigenoperator proportional to the identity, which is the steady state of this system.  The mode with the slowest relaxation for $\gamma>0$ is $\hat{\tau}_{S}\coloneqq \hat{\tau}_1$, which relaxes with rate $r_{S}\coloneqq \mathrm{Re}(r_1)$.  The relaxation rate $r_S$ is proportional to $\gamma$, and we work to first order in $\gamma$~\cite{sels_avalanche,alan_huse_avalanche}.

\emph{\textbf{Relaxation of the slowest mode---}}
\begin{figure}
    \centering
	\includegraphics[width=0.5\textwidth]{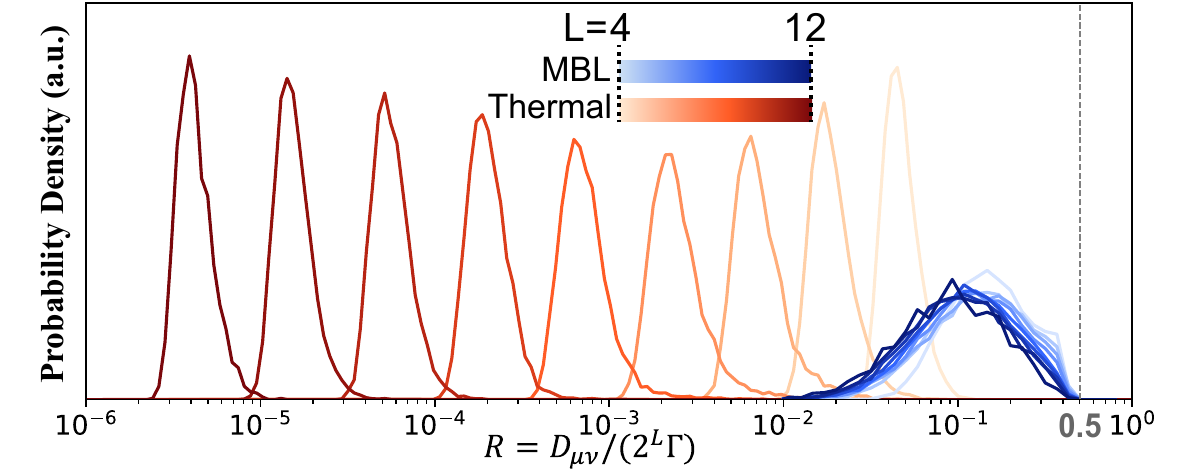}
\caption{
Probability distribution over samples of the ratio $R=D_{\mu\nu} / (2^L \Gamma)$ (see text). In the thermal regime (dark red to yellow), the ratio exponentially decays with increasing $L$. In comparison, $R$ barely drifts with $L$ for the MBL case (blue). We observe that $R$ never exceeds the value 0.5 (gray dashed line), which is explained with the minimal model in the main text. In this figure, we used $\alpha=30$ and $\alpha=1$ for MBL and thermal regimes, respectively.
}
	\label{fig:decay_ratio}
\end{figure}
In the weak-coupling limit, one can obtain $\hat{\tau}_{S}$ as a superposition of the diagonal terms $\ket{n}\bra{n}$, where $\ket{n}$ are the eigenstates of the closed system such that $\hat{U}_F \ket{n} = e^{i\theta_n}\ket{n}$. When $\gamma=0$, then $\ket{m}\bra{n}$ are eigenoperators of $S_{\mathrm{period}}$ with eigenvalues $e^{i(\theta_m-\theta_n)}$, so all diagonal terms $\ket{n}\bra{n}$ are degenerate, with $r_k=0$. Therefore, in the $\gamma \ll 1$ limit, one can obtain $\hat{\tau}_{S}$ in degenerate perturbation theory by diagonalizing the (super)operator $S[\hat{\rho}] \coloneqq \frac{1}{3}\sum_{j=1}^3 \hat{E}_j \hat{\rho} \hat{E}_j^\dagger$ in this degenerate subspace~\cite{sels_avalanche}, where the matrix elements are
\begin{align}
\label{eq:Snm}
S_{mn} &= \bra{m}S[\ket{n}\bra{n}]\ket{m} 
= \frac{1}{3}\sum_{j=1}^3 |\bra{m}\hat{E}_j \ket{n}|^2.
\end{align}
Note that this is a symmetric stochastic matrix with real eigenvalues. In particular, $\hat{\tau}_{S} = \sum_n c_n \ket{n}\bra{n}$ where $\vec{c}$ is the eigenvector of $S$ with the smallest spectral gap $\Gamma$ from the steady state [$\sum_m S_{nm} c_m = (1-\Gamma) c_n$]. The relaxation rate is $r_{S}=3\gamma\Gamma$. We normalize $\hat{\tau}_S$ so $\mathrm{Tr}\{\hat{\tau}_S^2\}=\sum_n|c_n|^2=2^L$.

Naively, the slowest mode is the local integral of motion (LIOM) that is farthest from the bath.  More precisely $\hat{\tau}_S$ is a traceless superposition of projectors on to the eigenstates of the closed system.  Among such operators, it is the one with the smallest weight of Pauli strings with nonidentity at site $L$~\cite{supp_mat}.  It is a LIOM that is indeed localized far from the bath, but it is different in detail from the $\ell$-bits and LIOMs of Refs.~\cite{serbyn_abanin_LIOM,huse_Oganesyan_lbit,chandran_abanin_liom_construction,Pekker-Refael2017_ww_flow}.

The latest time dynamics are determined by $\hat{\tau}_{S}$, with $\hat{\rho}(t) \simeq \hat{\mathbb{I}}/2^L + p_S e^{-r_{S}t}\hat{\tau}_{S}$ for any initial conditions that contain $\hat{\tau}_S$.  This assumes a nonzero gap between $(r_1/\gamma)$ and $(r_2/\gamma)$, which is the case for all samples examined. 
We can view the slow relaxation of $\hat{\tau}_S$ in terms of probability currents that flow between the eigenstates of the isolated system, leading to the final $\hat{\rho}= \hat{\mathbb{I}}/2^L$ equilibrium where all eigenstates have equal weight. 
Specifically, we may quantify the contribution $D_{mn}$ of the pair of eigenstates $m$, $n$ to the relaxation of $\hat{\tau}_S$ as
\begin{align}
D_{mn}\coloneqq S_{mn} (c_m-c_n)^2\geq 0.
\label{eq:Dnm}
\end{align}
One can show~\cite{supp_mat} that the relaxation rate of $\hat{\tau}_S$ is given by the sum of the contributions from all pairs of eigenstates of the closed system:
\begin{align}
\label{eq:sum_Dnm}
\sum_{m<n}D_{mn}=\Gamma~{\rm Tr}\{\hat{\tau}_S^2\} = 2^{L}\Gamma.
\end{align}
We find the pair of eigenstates $\ket{\mu}$,  $\ket{\nu}$ that gives the strongest contribution to the relaxation of $\hat{\tau}_S$ by finding the pair with the largest $D_{mn}$.  We quantify the fraction of the full relaxation that is due to this pair by the ratio $R \equiv {D_{\mu\nu}}/(2^L\Gamma)$.  The distribution of $R$ over disorder realizations and for different system sizes is shown in Fig.~\ref{fig:decay_ratio}, both for the thermal regime and for the MBL regime near the avalanche instability.  In the thermal regime, $R$ decreases exponentially with increasing $L$, indicating that many pairs of eigenstates are contributing similar amounts to the relaxation of the slowest mode, which is as should be expected for this thermal regime.

In the MBL regime, we find that this pair of eigenstates contributes an order-1 fraction that does not decrease substantially with increasing $L$.  This is consistent with previous works~\cite{serbyn_abanin_criterion,alan_huse_avalanche,garratt_roy_resonance,long_anushya_jacobi_resonance} that found extremely broad distributions for the matrix elements of local operators among MBL eigenstates.  On looking at the relaxation more thoroughly, we find that in this MBL regime the strongest contribution to the relaxation actually involves a set of four eigenstates, at least two of which are involved in a many-body near-resonance, as we will now describe in more detail.  

\begin{figure}
    \centering
	\includegraphics[width=0.5\textwidth]{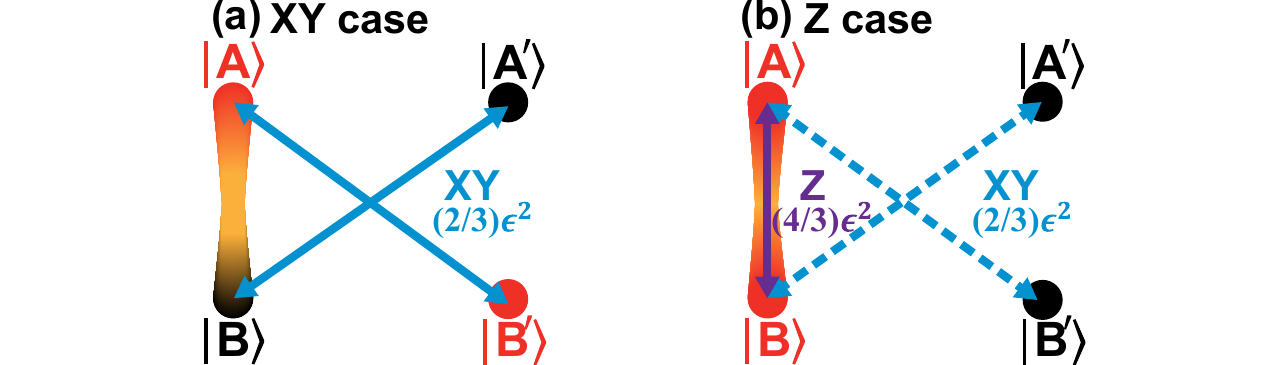}
	\caption{
Two distinct cases for the set of four most important eigenstates of $\hat{U}_F$ for the avalanche to propagate.
The near-resonant pair $\{A,B\}$ are coupled with $\hat{X}_L$ and $\hat{Y}_L$ jump operators to $B'$ and $A'$, respectively, i.e., they have anomalously large matrix elements.
In the $Z$ case (b), the near-resonant pair involves a spin flip next to the bath (site $L$), and the bath can additionally couple the resonant pair with the $\hat{Z}_L$ jump operator with matrix element $S_{AB}$ twice as large as $S_{A'B}$ and $S_{AB'}$. The two states with the largest $D_{\mu\nu}$ are colored red. In the $XY$ case (a), this pair is not the near-resonant pair, and $\{\mu,\nu\}=\{A,B'\}$; while they are the same in the $Z$ case, so $\{\mu,\nu\}=\{A,B\}$. In some samples, $\{A',B'\}$ are also near-resonant (not shown).}
	\label{fig:mechanism}
\end{figure}

\emph{\textbf{Near-resonant eigenstate set---}}
Deep in the MBL regime, we typically find that the  strongest contributions to the relaxation of $\hat{\tau}_S$ come from a set of four eigenstates, which include a near-resonant pair $\ket{A}$, $\ket{B}$ (we explain what ``near-resonant" means below).  The other two eigenstates involved are obtained by ``flipping'' the $\ell$-bit $\hat{\tau}_L$ adjacent to the bath: 
\begin{align}
    \ket{A'}=\hat{\tau}_L^x\ket{A},\quad \ket{B'}=\hat{\tau}_L^x\ket{B}.
\end{align}
In some cases, $\ket{A'}$ and $\ket{B'}$ are also near-resonant, as we discuss in the SM~\cite{supp_mat} with details.
The polarization $c_A=\langle A|\hat{\tau}_S\ket{A}$ of the slow mode differs substantially between the ``$A$'' eigenstates ($\ket{A}$ and $\ket{A'}$) and the ``$B$'' eigenstates, while it is essentially unchanged by flipping $\hat{\tau}_L$, so $c_A\cong c_{A'}$ and $c_B\cong c_{B'}$.  Thus it is the transitions between $\{A,A'\}$ eigenstates and $\{B,B'\}$ eigenstates that relax $\hat{\tau}_S$.  If we ``demix''~\cite{alan_huse_avalanche} the near-resonance to make a more localized (less entangled, more polarized) pair of orthonormal states $\ket{a}$, $\ket{b}$, we obtain
\begin{align}
    \ket{A} \cong \ket{a}-\epsilon\ket{b},\quad\ket{B}\cong\ket{b}+\epsilon \ket{a},
    \label{eq:ABab}
\end{align}
where $\ket{a}$ and $\ket{b}$ differ by spin flips at a nonzero fraction of all sites, including the sites that are most polarized in $\hat{\tau}_S$, 
so this is a long-range many-body resonance, in that sense, and it includes a ``flip'' of $\hat{\tau}_S$.  The spin that is flipped between $\ket{a}$ and $\ket{b}$ that is closest to the bath is at site $x$. For the largest $L$ that we can access, the most probable location of $x$ is near $x/L=0.8$~\cite{supp_mat}.

In the regime near the estimated avalanche instability, the ``mixing'' $\epsilon$ in the near-resonance is typically exponentially small in $L$, although it is exponentially larger than the mixing between other typical pairs of eigenstates that differ over a similar distance range. This is why we say ``near-resonance'': the mixing is relatively strong, partly due to the two eigenstates being near degeneracy in the spectrum of $\hat{U_F}$, but it is not fully resonant, since $\epsilon\ll 1$.

In many samples, $\{A',B'\}$ is not near-resonant and these states are well-approximated by $\ket{A'(B')}\cong\hat{X}_L\ket{a(b)}$. As a consequence of the near-resonance between $\{A,B\}$, i.e., the relatively large $\epsilon$ compared to other pairs of states, there are anomalously large matrix elements $|\bra{A(B)}\hat{E}_{1,2}\ket{B'(A')}|^2 \cong \epsilon^2$ of the two jump operators $\hat{E}_{1(2)}=\hat{X}_L(\hat{Y}_L)$. These drive two particularly large contributions, $D_{AB'}$ and $D_{BA'}$ (and their transposes), to the total relaxation rate of $\hat{\tau}_S$ with $S_{AB'} \cong S_{BA'} \cong 2\epsilon^2/3$. 

In addition, when the near-resonance involves flipping the polarization of site $L$ next to the bath (so $x=L$), there is another anomalously large matrix element. This time it is the matrix element $|\bra{A}\hat{Z}_L\ket{B}|^2\cong 4\epsilon^2$ of the jump operator $\hat{E}_3 = \hat{Z}_L$ between the two near-resonant eigenstates themselves. This results in $S_{AB} \cong 4\epsilon^2/3$ and $D_{AB} \cong 2 D_{AB'(BA')}$ (and its transpose).

Thus, there are two cases for $D_{\mu\nu}$, the largest contribution to Eq.~(\ref{eq:sum_Dnm}), corresponding to whether or not the resonance involves flipping the spin nearest to the bath. The bath can drive relaxation of the slowest mode through a resonance indirectly (with $X$ and $Y$ jump operators), i.e., the pairs of eigenstates involved are not the near-resonant pair, and sometimes directly (with $Z$ jump operator). We depict the two cases in Figs.~\ref{fig:mechanism}(a) and \ref{fig:mechanism}(b) and call them the ``$XY$" and ``$Z$" cases. Our claim is that the relaxation of the slowest mode, and thus the avalanche, proceeds via these dominant processes involving four eigenstates that include a particularly strong near-resonance, as explained in this section. 
This structure implies that the largest contribution $D_{\mu\nu}$ comes from \textit{either} the $Z$ jump operator or the $X$ and $Y$ jump operators of the bath, depending on which of the above cases is relevant. In the $Z$ case, which is $x=L$, the largest contribution ($A$-$B$ pair) is accompanied by two $XY$ contributions of half the magnitude ($A$-$B'$ and $B$-$A'$ pairs). In the $XY$ case, which is $x<L$, there are two roughly equal largest contributions ($A$-$B'$ and $B$-$A'$). In both cases, the largest contribution $D_{\mu\nu}$ doesn't exceed half of the total relaxation.
The phenomenology of the near-resonant eigenstate set explained in this section is corroborated by numerical observations in the next section. An extension of this simple description is discussed in the SM~\cite{supp_mat}.

\begin{figure}
    \centering
	\includegraphics[width=0.5\textwidth]{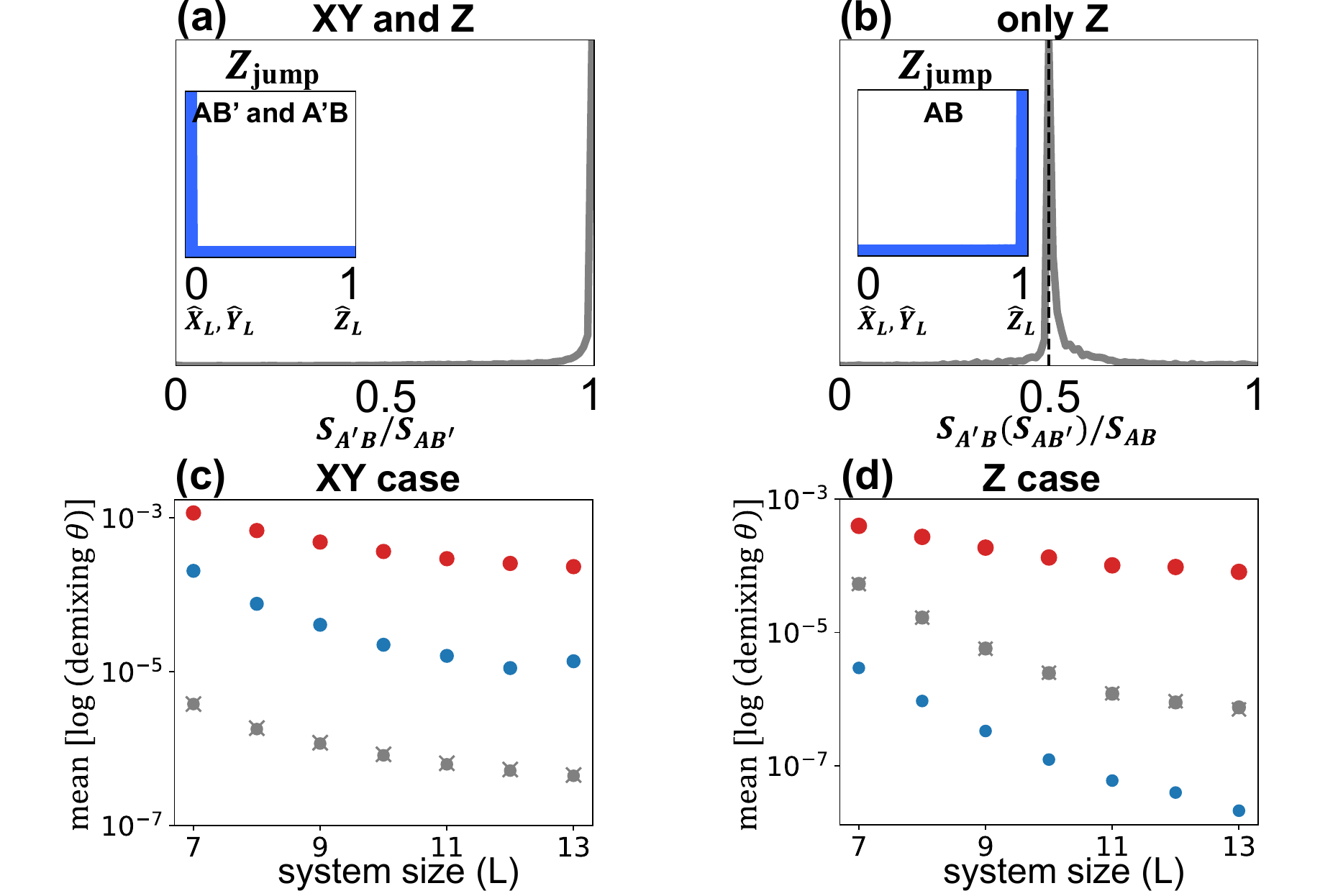}
	\caption{  Consistency with the near-resonant eigenstate set.  All data here are for $\alpha=30$.
 (a) The near-resonance between $\{A,B\}$ causes the $AB'$ and $A'B$ matrix elements to be roughly equal, as described in the minimal model. The distribution of $S_{AB'}/S_{A'B}$ is indeed peaked near $1$. The inset shows $Z_{\mathrm{jump}}$ of $AB'$ and $A'B$ are peaked at zero (coupled with $\hat{X}_L$ and $\hat{Y}_L$).
 (b) When the near-resonant pair involves the spin flip of the last site ($Z$ case), the bath can directly couple $\{A,B\}$. The distribution of the shown ratio of matrix elements demonstrates that $S_{AB}\simeq 2S_{AB'}\simeq 2S_{A'B}$. The inset shows $Z_{\mathrm{jump}}$ of $AB$ is peaked at 1, implying $\{A,B\}$ are coupled with $\hat{Z}_L$.
 The bottom two panels show the demixing angles for (c) the $XY$ case and (d) the $Z$ case. The points correspond to $AB$ (red), $A'B$ (gray $\times$), $AB'$ (gray dot), and $A'B'$ (blue). The near-resonant pair identified as in the main text ($AB$) has the largest demixing angles for all cases. 
 The calculations in (a),(b) are done with $L=11$ and $10^4$ disorder realizations. }
	\label{fig:numerics}
\end{figure}

\emph{\textbf{Numerical Observations---}}
We numerically determine the eigenstates of $\hat{U}_F$, the slowest mode of $S$, and the contributions $D_{mn}$ to its relaxation rate associated with probability currents between eigenstates. We examine the pairs of eigenstates that contribute the most ($\ket{\mu}$ and $\ket{\nu}$), identify the four important eigenstates discussed earlier, and quantify how they are related to each other to verify that the near-resonant set of eigenstates does indeed have that structure.

Recall that in our model, the polarization has a preferred spin direction: $Z$. Deep in the MBL regime, the local $\ell$-bit operators typically are very close to the local single-spin Pauli operators.
Therefore, we identify $\ket{\gamma'}\equiv\tau^x_L\ket{\gamma}$ for $\gamma\in\{\mu,\nu\}$ simply by finding the eigenstate with largest $\lvert\bra{\gamma}\hat{X}_L\ket{n}\rvert$ among all eigenstates $\ket{n}$.

To determine which jump operator from the bath dominantly couples $\ket{\mu}$ and $\ket{\nu}$ we define a tool $Z_{\mathrm{jump}}(m,n) \coloneqq |\bra{m}\hat{Z}_L \ket{n}|^2 / \sum_{j=1}^3|\bra{m}\hat{E}_j \ket{n}|^2$.
% \begin{align}
%     Z_{\mathrm{jump}}(m,n) \coloneqq \frac{|\bra{m}\hat{Z}_L \ket{n}|^2}{\sum_{j=1}^3|\bra{m}\hat{E}_j \ket{n}|^2}.
% \end{align}
Note that $\hat{E}_3 = \hat{Z}_L$. We find that $Z_{\mathrm{jump}}(\mu,\nu)$ is extremely close to 0 or 1 in any one sample, as shown in the inset of Fig.~\ref{fig:numerics}(a) and \ref{fig:numerics}(b). This separation is captured by the minimal model since $Z_{\mathrm{jump}}(\mu,\nu)\cong 0$ (or $1$) is associated to the $XY$ (or $Z$) case where the pair $\{\mu, \nu\}$ isn't (or is) the near-resonant pair $\{A, B\}$. 
We therefore identify the near-resonant pair for each sample based on the minimal model as follows.
In the case that $\ket{\mu}$ and $\ket{\nu}$---the pair with the largest $D_{mn}$---are ``connected" with $Z$ ($Z_\mathrm{jump} \cong 1$), these states are identified as $\ket{A}$ and $\ket{B}$.
Otherwise, if they are connected with $X$ and $Y$ instead ($Z_\mathrm{jump} \cong 0$), we associate $\{\ket{A}, \ket{B})\}$ to $\{\ket{\mu'}, \ket{\nu}\}$ or $\{\ket{\mu}, \ket{\nu'}\}$, whichever pair has the smaller quasienergy splitting.

We first checked that $c_\gamma \cong c_{\gamma'}$ indeed holds for $\gamma\in\{A,B\}$. This relation is satisfied because $S_{\gamma'\gamma}$ is of order one, so the slow mode $\hat\tau_S$ contains negligible population difference between $\gamma$ and $\gamma'$. 
Also, we checked that $|c_\mu - c_\nu|$ is of order one, which should be true as we picked the pair with the largest $D_{\mu\nu}= S_{\mu\nu} (c_\mu-c_\nu)^2$.
Therefore, the matrix elements $S_{mn}$ (Eq.~\ref{eq:Snm}) are a good proxy for $D_{mn}$ (Eq.~\ref{eq:Dnm}) within this set of important eigenstates. We compare these matrix elements to confirm the results are consistent with the minimal model described in Fig.~\ref{fig:mechanism}(a) and \ref{fig:mechanism}(b). As we show in Fig.~\ref{fig:numerics}(a), $S_{AB'} \cong S_{BA'}$ holds (and their transposes), and $XY$ jump operators couple them (inset). Furthermore, when $x=L$,  we additionally observe $S_{AB} \cong 2 S_{AB'} \cong 2 S_{BA'}$ as presented in Fig.~\ref{fig:numerics}(b) and $\{A,B\}$ is coupled with $\hat{Z}_L$ (inset).
We note that both cases satisfy $R<1/2$.

We finally tested our picture using the ``demixing" procedure of Ref.~\cite{alan_huse_avalanche}, calculating the most localized basis of a subspace spanned by \textit{two} eigenstates, and the corresponding basis rotation. The basis rotation corresponds to a location on a Bloch sphere; the polar angle, called the ``demixing angle," is a measure of the resonance strength between the two eigenstates [$\epsilon$ in Eq.~(\ref{eq:ABab})].
As shown in Figs.~\ref{fig:numerics}(c) and \ref{fig:numerics}(d), the demixing angle between the eigenstates we identified as $\{A,B\}$ is by far the largest among other pairs in our eigenstate set, consistent with the idea that that pair is indeed a near-resonance that enables the bath to relax the slowest mode, i.e., the avalanche to propagate. 

\emph{\textbf{Conclusion---}}
Our work bridges the two key ingredients that destabilize many-body localization---\textit{many-body resonances} and the \textit{avalanche instability}---demonstrating that avalanches spread primarily by strong rare near-resonances.
Assuming the avalanche has proceeded for a sufficiently large distance, we model the putative thermal bubble as an infinite Markovian bath with infinite temperature. In this model, we discovered the existence of dominant processes in the avalanche, involving only a few pairs of eigenstates of the closed system, including a strong near-resonance. The avalanche proceeds through these rare eigenstate pairs, leveraging many-body near-resonances to relax the spins some distance away along the chain. 
The inner structure of the dominant set of eigenstates is dictated by whether or not the associated resonance involves flipping a spin at the site next to the bath. This sets what jump operators effectively use the resonance present in the closed system to spread the avalanche. We presented a minimal model involving two near-resonant eigenstates and two additional auxiliary states to explain how this works in detail, and verified our picture with numerical observations. Our work advances the understanding of the avalanche instability of many-body localization and provides a detailed connection to rare many-body resonances present in MBL systems. Some further discussion of our conclusions and model assumptions is provided in the SM~\cite{supp_mat}.

\begin{acknowledgments}
\emph{\textbf{Acknowledgements---}}
We thank Sarang Gopalakrishnan, Vedika Khemani, Jacob Lin, and Vir Bulchandani for discussions, and Luis Colmenarez and David Luitz for previous collaboration on related work.  The work at Princeton was supported in part by NSF QLCI grant OMA-2120757.  A.M. was also supported in part by the Stanford Q-FARM Bloch Postdoctoral Fellowship in Quantum Science and Engineering and the Gordon and Betty Moore Foundation’s EPiQS Initiative through Grant No. GBMF8686.
Simulations presented in this work were performed on computational resources managed and supported by Princeton Research Computing.
\end{acknowledgments}

\bibliography{main}

\end{document}

\title{\textbf{Supplemental Material for \\
``Many-body resonances in the avalanche instability of many-body localization''}}

\author{Hyunsoo Ha}
\affiliation{Department of Physics, Princeton University, Princeton, New Jersey 08544, USA}

\author{Alan Morningstar}
\affiliation{Department of Physics, Princeton University, Princeton, New Jersey 08544, USA}
\affiliation{Department of Physics, Stanford University, Stanford, California 94305, USA}

\author{David A. Huse}
\affiliation{Department of Physics, Princeton University, Princeton, New Jersey 08544, USA}

\date{\today}
\maketitle
\onecolumngrid

\begin{center}
\begin{minipage}{0.85\textwidth}
\vspace{-0.8cm}
In this supplemental material, we (i) provide details on the model for the closed system that hosts thermal to MBL transition, (ii) derive useful relations on the slowest mode, and quantitatively compare it with $\ell$-bits and LIOMs introduced in the previous references, (iii) give details on the decay rate of the density matrix at the latest time, which is governed by the slowest mode, (iv) investigate the effect of flipping $\hat{\tau}_L$ on detuning of the near-resonance, (v) present the range of near-resonance that was prominently used in the slowest mode, and (vi) discuss the justification of the small-system-weakly-coupled-to-a-bath avalanche model that we use.
\vspace{1cm}
\end{minipage}
\end{center}

%%%%%%%%%%%%%%%%%%%%%%%%%%%%%%%%%%%%%%%%%%%%%%%%%%%%%%%%%%%%%%%%%
\setcounter{section}{0}
\setcounter{figure}{0}
\setcounter{equation}{0}
\renewcommand{\thefigure}{S\arabic{figure}}
\renewcommand{\theequation}{S\arabic{equation}}
\renewcommand{\thesection}{S\arabic{section}}

\section{(i) Random Floquet Model}
\label{model}
We use the Floquet circuit model of Morningstar, \textit{et al.}~\cite{alan_huse_avalanche} and depicted in Fig.~\ref{fig:randomfloquet}. The Floquet unitary operator $U_F$ is obtained by applying a layer of one-site gates $U_d$ followed by $L-1$ two-site gates---one per bond---denoted by $U_u$. The onsite gates are first sampled from the $2\times2$ circular unitary ensemble (CUE) and then diagonalized, so the localizing disorder of the model is in the $Z$ direction. The two-site gates $u_i$ act on sites $i$ and $i+1$ and can be written as
\begin{align}
    u_i = \mathrm{exp}\left( \frac{i}{\alpha}M_i \right)
\end{align}
where $M_i$ are sampled from the $4\times4$ Guassian unitary ensemble (GUE). $1/\alpha$ controls the interaction strength so that the effect of the disorder gets stronger as $\alpha$ increases. The order of applying the $u_i$'s is determined randomly for each sample.
\begin{figure}[h]
\centering
\includegraphics[width=0.5\textwidth]{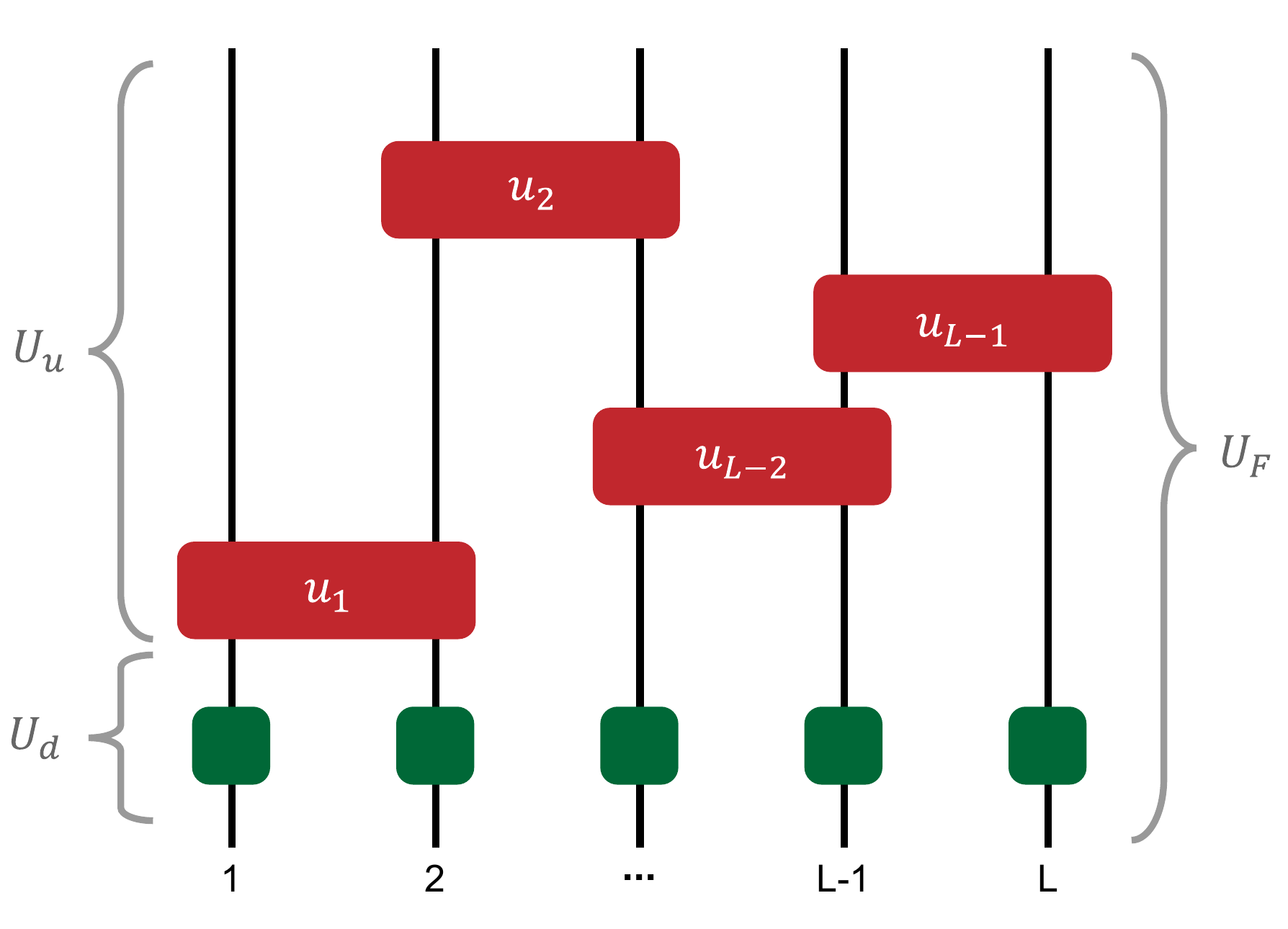}
\caption{ 
 Random Floquet Model introduced in Ref.~\cite{alan_huse_avalanche}. Details of the model are described in the Appendix. }  
\label{fig:randomfloquet}
\end{figure}

\section{(ii) Slowest Mode}
As explained in the main text, one can obtain the slowest mode $\hat{\tau}_S$ in the weak-coupling limit by diagonalizing the superoperator $S[\hat{\rho}] \coloneqq \frac{1}{3}\sum_\mu \hat{E}_\mu \hat{\rho} \hat{E}_\mu^\dagger$, in the degenerate subspace of Floquet eigenstates $\ket{n}\bra{n}$, where the matrix elements are
\begin{align}
S_{nm} = \bra{m}S[\ket{n}\bra{n}]\ket{m} = \frac{1}{3}\sum_\mu |\bra{m}\hat{E}_\mu \ket{n}|^2
\end{align}
and $\hat{E}_\mu = (\hat{X}_L, \hat{Y}_L, \hat{Z}_L)$ are the jump operators at the spin next to the bath.
In particular, $\hat{\tau}_{S} = \sum_n c_n \ket{n}\bra{n}$ where $\vec{c}_n$ is the eigenvector of $S_{nm}$ with the smallest spectral gap $\Gamma$.
As the slowest mode is constructed from the degenerate subspace of $\ket{n}\bra{n}$, it is a conserved quantity (a LIOM) of the purely unitary dynamics given by $U_F$.

We derive useful features of the slowest mode $\hat{\tau}_S$ below.

\begin{lemma}
For any Pauli-string operator $P'$ acting on the $(L-1)$ spins away from the bath, and $\hat{E}_\mu = (\hat{X}_L, \hat{Y}_L, \hat{Z}_L)$ acting on the last spin next to the bath, the following holds:
\begin{align}
\label{eq:lemma_S[PS]}
S[P' \otimes \mathbb{I}_L] = P' \otimes \mathbb{I}_L, \quad S[P' \otimes \hat{E}_{\mu}] = -\frac{1}{3}P' \otimes \hat{E}_{\mu}.
\end{align}
\label{lemma:S}
\end{lemma}

\begin{proof}
The super-operator $S$ only acts on site $L$, such that for an arbitrary matrix $M\in \mathbb{M}_{(2,2)}$, $S[P' \otimes M] = P' \otimes \left(\frac{1}{3}\sum_{\mu} \hat{E}_{\mu} M \hat{E}_{\mu}^\dagger \right)$. Eq.~\ref{eq:lemma_S[PS]} follows by substituting $\mathbb{I}_L$ and $\hat{E}_\mu$ in $M$.
\end{proof}

\begin{theorem}
\label{theorem:pauli}
The slowest mode is the conserved quantity of the closed system (which is not an identity) with the smallest weight from the Pauli strings with non-identity at site $L$ (connected to the bath).
\end{theorem}

\begin{proof}
Consider a conserved quantity written as $\hat{\tau} = \sum_n c_n \ket{n}\bra{n}$ where $c_n\in \mathbb{R}$. The normalization condition is $\sum_n c_n^2=2^L$ as we defined in the main text. We can write the slowest mode in the Pauli-string basis like $\hat{\tau}_{S}=\sum_P \bar{c}_{P} P$ with $P$ being the Pauli-strings. The coefficients are $\bar{c}_P={\rm Tr}(\hat{\tau}_{S}P)/2^L$. In the Pauli-string basis, the normalization condition is rewritten as $\sum_P \bar{c}_P^2 = 1$. The total weight of the Pauli-strings with non-identity at site $L$ is
\begin{align}
\label{eq:L_weight}
    w_L &= \sum_{P'} \bar{c}_{P'\otimes \hat{X}_L}^2 + \sum_{P'} \bar{c}_{P'\otimes \hat{Y}_L}^2 + \sum_{P'} \bar{c}_{P'\otimes \hat{Z}_L}^2 = \sum_{\mu} \sum_{P'} \bar{c}_{P' \otimes \hat{E}_\mu}^2
\end{align}
where $P'$ is a Pauli-string defined on L-1 sites and $\hat{E}_{\mu} = (\hat{X}_L, \hat{Y}_L, \hat{Z}_L)$ are the Pauli operators from site $L$.

Next, we define a functional $R(\vec{c}) \coloneqq \sum_{n,m}c_n S_{nm} c_m$. Below, we show how the functional $R$ is related to $w_L$.
\begin{align}
{\rm Tr}\left( \hat{\tau} S[\hat{\tau}])\right) &= {\rm Tr}\left( \sum_l c_l \ket{l}\bra{l} ~\sum_n c_n S[\ket{n}\bra{n}]\right) \nonumber\\
&= {\rm Tr}\left( \sum_l c_l \ket{l}\bra{l} ~\sum_{n,m} c_n S_{nm} \ket{m}\bra{m}\right) \nonumber\\
&= \sum_{n,m}c_n S_{nm} c_m = R
\end{align}
We can rewrite the same equation under the Pauli-string basis, $\hat{\tau} = \sum_P \bar{c}_P P$. Here, we use Lemma~$\ref{lemma:S}$ introduced earlier.
\begin{align}
{\rm Tr}\left( \hat{\tau} S[\hat{\tau}])\right) &= \sum_{P'} \bar{c}_{P' \otimes \mathbb{I}}^2 - \frac{1}{3}\sum_\mu \sum_{P'} \bar{c}_{P' \otimes \hat{E}_\mu}^2 = 1 - \frac{4}{3} \sum_\mu \sum_{P'} \bar{c}_{P' \otimes \hat{E}_\mu}^2 = 1 - \frac{4}{3} w_L
\end{align}
From the above two equations, we derive the relation $R = 1 - (4/3) w_L$.

In the final step, let's define $f \coloneqq \sum_n c_n^2$. Using the Lagrange multiplier method, functional $R$ is local extremum under the constraint $f=2^L$ when $ \nabla R - \lambda \nabla f = 0$ holds, which is equivalent to the eigenvalue problem: $\forall n,\quad \sum_m S_{nm} c_m = \lambda c_n$. In this condition, we obtain $R=\lambda$, and the weight of the Pauli-strings with non-identity at site $L$ is $w_L = (3/4)(1-\lambda)$. Therefore, finding $\hat{\tau}$ which minimizes $w_L$ is equivalent to finding the eigenvector of the super-operator $S$ with the second largest eigenvalue (the largest eigenvalue is one which corresponds simply to an identity matrix). Hence, the conserved quantity with the smallest $w_L$ except the identity is the slowest mode $\hat{\tau}_{S}$.
\end{proof}
\clearpage

% \\
\begin{theorem}
The slowest mode is traceless.
\end{theorem}

\begin{proof}
From the definition of the slowest mode, it is an eigenmatrix of superoperator $S$ with an eigenvalue $1-\Gamma$. Therefore, ${\rm Tr}(\hat{\tau}_{S})={\rm Tr}(S[\hat{\tau}_{S}])/(1-\Gamma)={\rm Tr}(\frac{1}{3}\sum_\mu \hat{E}_\mu \hat{\tau}_{S} \hat{E}_\mu^\dagger)/(1-\Gamma)={\rm Tr}(\hat{\tau}_{S})/(1-\Gamma)$. As $0<\Gamma<1$, it should be traceless ${\rm Tr}(\hat{\tau}_{S})=0$.  
\end{proof}

We numerically compared the slowest mode $\hat{\tau}_S$ with the LIOM ($\tau_1$) and $\ell$-bit ($\bar{\tau}_1$) introduced in Refs.~\cite{serbyn_abanin_LIOM,huse_Oganesyan_lbit,chandran_abanin_liom_construction}. The explicit forms of $\tau_1$ and $\bar{\tau}_1$ are
\begin{align}
\tau_1 &\equiv \sum_n \bra{n}\hat{Z}_1\ket{n} \ket{n}\bra{n}\\
\bar{\tau}_1 &\equiv \sum_n \mathrm{sgn}(\bra{n}\hat{Z}_1\ket{n} )\ket{n}\bra{n},
\end{align}
where $\ket{n}$s are the eigenfunctions of the closed system. Note that $\tau_1$ and $\bar{\tau}_1$ are the conserved quantities that are localized near site $1$, farthest from the bath.
In particular, we calculate the Pauli string weight $w_i(\hat{\tau})$ for $\hat{\tau} \in \{\tau_1, \bar{\tau}_1,\hat{\tau}_S\}$, 
which is the total weight of the Pauli string bases with the last non-identity matrix placed at $i$-th site. These Pauli bases are products of identity matrices for sites farther than $i$.

We compare the Pauli string weight in Fig.~\ref{fig:LIOM_comparison}. For sufficiently large disorder strength $\alpha$, all three operators are localized far from the bath. It is clear that the slowest mode is different from $\tau_1$ and $\bar{\tau}_1$. Compared to the other two quantities, the slowest mode has the least weight at the last site, as we proved in Theorem~\ref{theorem:pauli}. The difference manifests when we compare the three operators for a single sample. While $\tau_1$ and $\bar{\tau}_1$ always have maximum weight on the first site, the slowest mode often has maximum weight other than the first site but is still far away from the bath. What constrains the slowest mode is not the first site to have the maximum weight but to have the least weight on the last site next to the bath, which was indeed true for every sample.
\begin{figure*}
\centering
\includegraphics[width=\textwidth]{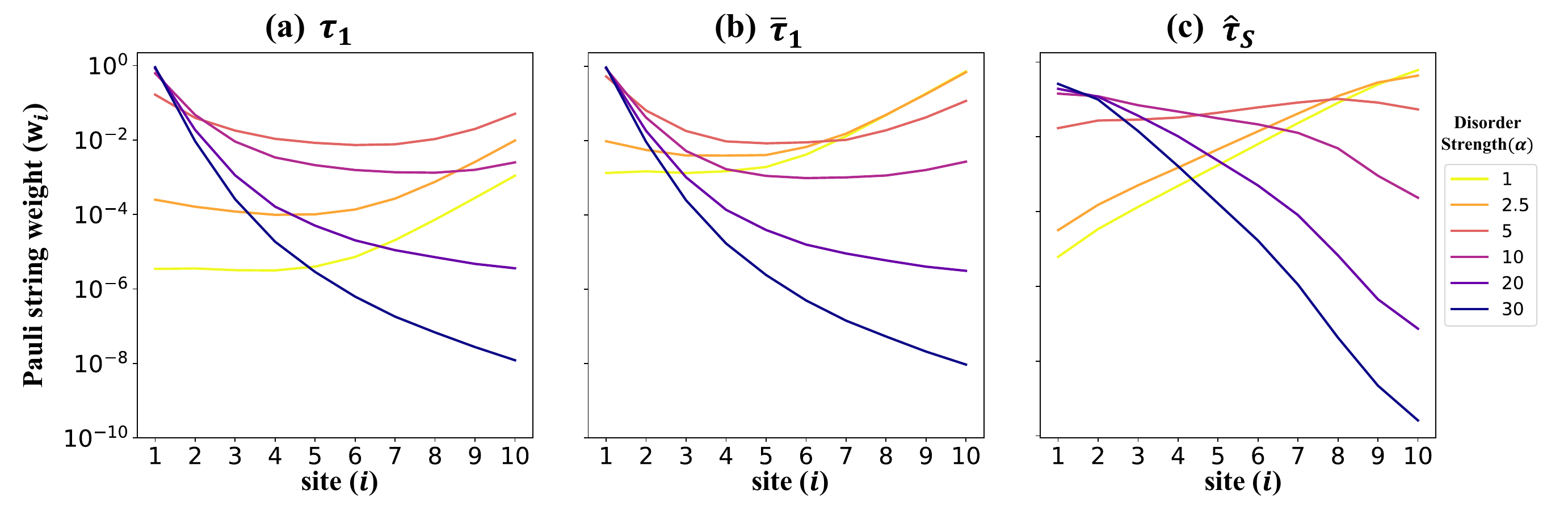}
\caption{ Pauli String weight $w_n$ for (a) $\tau_1$, (b) $\bar{\tau}_1$, and (c) the slowest mode $\tau_S$, averaged (after taking the logarithm) for $10^3$ samples with system size $L=10$.
 }  
\label{fig:LIOM_comparison}
\end{figure*}

\section{(iii) Derivation of Equation 5}
As explained in the main text, the slowest mode is $\hat{\tau}_S=\sum_n c_n \ket{n}\bra{n}$ where $\vec{c}_n$ is the eigenvector of a symmetric Markovian matrix $S_{nm}$ with the smallest spectral gap $\Gamma$ ($\sum_m S_{nm} c_m = (1-\Gamma) c_n$). The operator $\hat{\tau}_S$ is normalized such that $\lVert \hat{\tau}_S \rVert_1 ^2=\sum_n c_n^2 = 2^{L}$.  In our model, the relaxation rate of the slowest mode is $r_{S}=3\gamma\Gamma$. 
We assume a gap between the smallest and second-smallest spectral gaps. With this assumption, the density matrix at the latest time is asymptotically $\hat{\rho}(t) \simeq (\hat{\mathbb{I}} + p_S e^{-r_{S}t}\hat{\tau}_{S})/2^L$. Therefore, $\hat{\rho}(t)$ thermalizes and converges to $\hat{\rho}(\infty)=\mathbb{I}/2^L$ asymptotically at the latest time with the rate:

\begin{align}
\begin{split}
\frac{d}{dt}\left(\lVert \hat{\rho}(t)-\hat{\rho}(\infty) \rVert_1\right)^2 
&\simeq\frac{d}{dt}\left(\lVert  2^{-L} p_S e^{-r_S t} \hat{\tau}_S \rVert_1\right)^2 =\frac{d}{dt} \left(2^{-2L}|p_S|^2~ e^{-2r_S t}~ \lVert \hat{\tau}_S \rVert_1 ^2\right) =\frac{d}{dt} \left(|p_S|^2~ e^{-2r_S t}~ 2^{-L}\right)  \\
&=-2r_S~|p_S|^2~ e^{-2r_S t} ~ 2^{-L} =-3\gamma~|p_S|^2~ e^{-2r_S t} ~ 2^{-L+1}~ \Gamma
\end{split}
\end{align}

Also, one can view $\hat{\tau}_S$ as specifying a structure of probability currents that flow between eigenstates of the isolated system during the late-time approach to equilibrium.
Specifically, if we define a current $j_{nm}\coloneqq S_{nm} (c_n-c_m)$, then $\frac{d}{dt}\hat{\rho}_{mm}(t) = 3p_S\gamma2^{-L} e^{-r_{S}t}\sum_n j_{nm}$ holds at late times, which implies that $j_{nm}$ obtained from $\hat{\tau}_{S}$ determines the structure of the late time thermalization. One can estimate how $\hat{\rho}(t)$ converges to $\hat{\rho}(\infty)$ with the inner structure of $\hat{\tau}_S$:
\begin{align}
\frac{d}{dt}\left(\lVert \hat{\rho}(t)-\hat{\rho}(\infty) \rVert_1\right)^2 &\simeq \frac{d}{dt} \left(\sum_m |\hat{\rho}_{mm} (t)-1/2^L|^2 \right) = 6\gamma|p_S|^22^{-2L}\sum_m e^{-r_S t} c_n \left(  e^{-r_{S}t}\sum_n j_{nm} \right)\nonumber\\
&=  3\gamma|p_S|^2 2^{-2L} e^{-2r_{S}t}\sum_{n,m}(c_n-c_m) j_{mn} = -3\gamma|p_S|^2 2^{-2L} e^{-2r_{S}t}\sum_{n,m}(c_n-c_m)^2 S_{nm} \\
&= -3\gamma|p_S|^2 2^{-2L} e^{-2r_{S}t}\sum_{n,m}D_{nm}\nonumber
\end{align}
Here, we used the fact that at the latest time, only the diagonal terms in $\hat{\rho}(t)$ remain (under the weak coupling limit $\gamma\ll 1$, the slowest mode consists of only the diagonal terms). $D_{nm}$ quantifies the contribution of a pair $\ket{n}$ and $\ket{m}$ to the relaxation of the slowest mode.

Comparing the two equations above, one can derive the following:
\begin{align}
    \sum_{n,m}D_{nm}=\sum_{n,m}(c_n-c_m)^2S_{nm} = 2^{L+1}\Gamma.
\end{align}
This can also be directly derived by:
\begin{align}
    \sum_{n,m} D_{nm} &= \sum_{n,m} (c_n-c_m)^2 S_{nm} \nonumber \\ &=\sum_{n,m} (c_n^2 + c_m^2) S_{nm} - 2\sum_n c_n \left(\sum_m S_{nm}c_m\right) \\
    &=2^{L+1} - 2^{L+1}(1-\Gamma) =2^{L+1}\Gamma.\nonumber
\end{align}

From the above derivations, the largest contribution $D_{\mu\nu}$ is related with $\Gamma$ by
\begin{align}
    D_{\mu\nu} = R \sum_{n<m}D_{nm} = 2^{L} R \Gamma.
\end{align}
Therefore, one can naively derive the relation between the thermalization rate and the matrix element $S_{\mu\nu}$ between the most important pair of states with an effective O(1) value of $R$ as follows.
\begin{align}
    r_S = 3\gamma \Gamma = \frac{3\gamma(c_\mu - c_\nu)^2}{R}  2^{-L}S_{\mu\nu}
\end{align}
As the observed ratio $R$ and $c_\mu-c_\nu$ are both O(1) values for MBL regimes, we estimate a relation $r_S \sim 2^{-L} S_{\mu\nu}$.

\section{(iv) Detuning of near-resonance by flipping $\hat{\tau}_L$}
\begin{figure*}
\centering
\includegraphics[width=\textwidth]{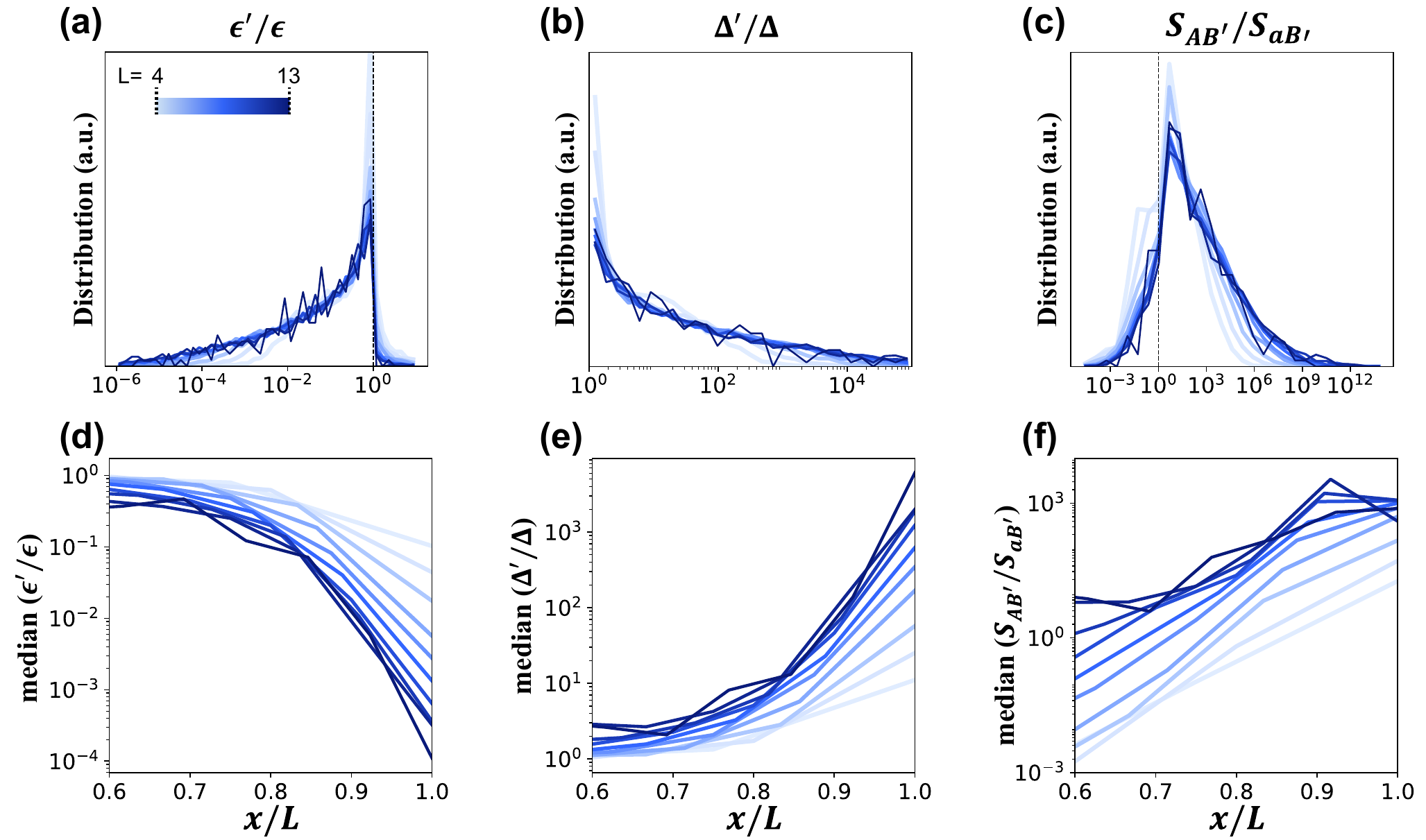}
\caption{Comparing ($A,B$) and ($A',B'$).  All data here are for $\alpha=30$ and $L\in[4,13]$. Distributions of (a) $\epsilon'/\epsilon$, (b) $\Delta'/\Delta$, and (c) $S_{AB'}/S_{aB'}$ show there are many samples satisfying $|\epsilon'|\ll|\epsilon|$ or $|\epsilon|\sim|\epsilon'|\sim|\epsilon-\epsilon'|$, but cases with $|\epsilon-\epsilon'|\ll|\epsilon|$ are less common (see text of the Appendix). The median of (d) $\epsilon'/\epsilon$, (e) $\Delta'/\Delta$, and (f) $S_{AB'}/S_{aB'}$ vs. the location of the last spin-flip (${x}/L$) of the near-resonance are illustrated. The detuning due to flipping $\hat{\tau}_L$ gets more significant as the last spin-flip gets closer to the bath.
}
\label{fig:prime}
\end{figure*}

In this section, we examine to what extent there is also a near-resonance between eigenstates $\ket{A'}$ and $\ket{B'}$.  These two eigenstates differ from the near-resonant states $\ket{A}$ and $\ket{B}$ by the flipping of $\ell$-bit $\hat{\tau}_L$.  This change may strongly detune the near-resonance, which is what happens for samples where the near-resonance involves flipping spins that are close to or include spin $L$. But in other samples, the near-resonance only involves spin flips that are rather far from spin $L$, so flipping  $\hat{\tau}_L$ has a rather small effect on the near-resonance. We denote the location of the spin-flip that is closest to $L$ as ${x}$. 

The minimal wavefunction model introduced in the main text (Eq.~7) neglected the possibility of a near-resonance between ($A', B'$).
Revising our minimal model by including and ``de-mixing'' the possible near-resonance between ($A', B'$), one can write:
\begin{equation}
\label{eq:revised_minimal}
\begin{split}
\ket{A} = \ket{a}-\epsilon \ket{b} &\quad \ket{B} = \ket{b}+\epsilon \ket{a}\\
\ket{A'} = \ket{a'}-\epsilon'\ket{b'} &\quad \ket{B'} = \ket{b'} + \epsilon' \ket{a'}.
\end{split}
\end{equation}
The quasi-energy splittings of these two near-resonances are $\Delta = |\theta_A-\theta_B|$ and $\Delta' = |\theta_{A'}-\theta_{B'}|$; we have labelled them so that $\Delta<\Delta'$.
We can partially quantify how much flipping $\hat{\tau}_L$ alters the near-resonance by comparing $\epsilon'$ and $\Delta'$ to $\epsilon$ and $\Delta$.

The distributions of $\epsilon'/\epsilon$ and $\Delta'/\Delta$ for different system sizes are shown in Fig.~\ref{fig:prime}(a,b). Both distributions show long tails to 
$|\epsilon'|\ll|\epsilon|$ and $\Delta'\gg\Delta$.  This shows that in a substantial fraction of the samples the near-resonance is strongly detuned by flipping $\hat{\tau}_L$. Fig.~\ref{fig:prime}(d,e) shows that this mostly happens due to the near-resonance extending to near site $L$, so $x/L$ is close to or equal to one.

There are also many samples that have  $|\epsilon|$, $|\epsilon'|$, and $|\epsilon-\epsilon'|$ all of the same order, as indicated by the peaks near $\epsilon'/\epsilon\cong\Delta'/\Delta\cong 1$ in Fig.~\ref{fig:prime}(a,b).  These are the cases where flipping $\hat{\tau}_L$ does alter the near-resonance, but not by a particularly large or small amount, which mostly happens for smaller ${x}/L$ [see Fig.~\ref{fig:prime}(d,e)].

The cases where flipping $\hat{\tau}_L$ has very little effect on the near resonance should have $\epsilon'/\epsilon$ and $\Delta'/\Delta$ very close to one, so would make a very sharp peak in those distributions, which is not seen in Fig.~\ref{fig:prime}(a,b).  This suggests that only a small fraction of samples are in this category.  For such samples $S_{AB'}$ is small due to a near cancellation between the contribution from $\ket{a}$ coupling to $\ket{a'}$ and that from $\ket{b}$ coupling to $\ket{b'}$.  This near-cancellation will not happen if we replace $\ket{A}$ with $\ket{a}$ and look instead at $S_{aB'}$. 
Thus in Fig.~\ref{fig:prime}(c) we show the distribution of $S_{AB'}/S_{aB'}$ ($\sim |(\epsilon-\epsilon’)/\epsilon'|^2$).
The much weaker tail in this distribution to very small $S_{AB'}/S_{aB'}$ are the samples where the near-resonance is only very weakly affected by flipping $\hat{\tau}_L$, while the much stronger tail to very large $S_{AB'}/S_{aB'}$ are those more common samples where this flip strongly detunes the near-resonance. 
Fig.~\ref{fig:prime}(f) shows the median $S_{AB'}/S_{aB'}$, with the expected trend of stronger detuning (larger $S_{AB'}/S_{aB'}$) as $x/L$ increases and the near-resonance thus extends closer to the flipped $\hat{\tau}_L$.

\section{(v) Range of Near-Resonance}
\begin{figure}
\centering
\includegraphics[width=0.5\textwidth]{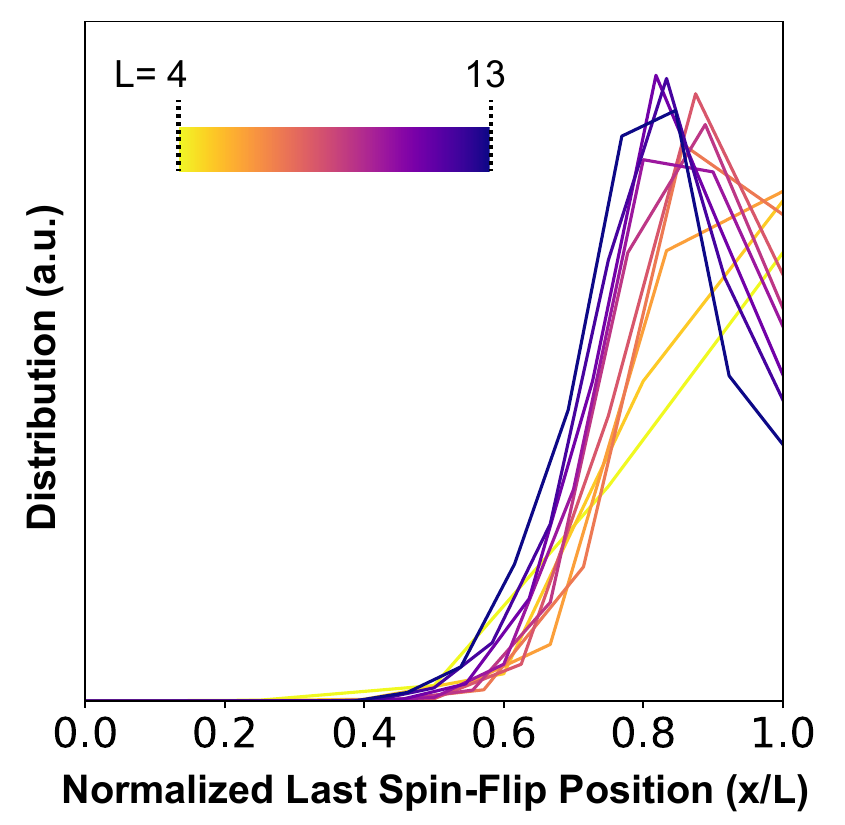}
\caption{
Distribution of the last spin flip location. We normalized the location with system size ($x/L$). All data here are for $\alpha=30$.}

\label{fig:Zn}
\end{figure}

The previous sections described how the bath uses near-resonances to flip the spins at the other end of an MBL chain and thermalize the system. Therefore, it is perhaps natural to expect the near-resonance involved to be an end-to-end resonance such that the last spin-flip between $\ket{A}$ and $\ket{B}$ should not be far away from the bath. Strikingly, however, the distance from the bath to the last spin-flip appears to increase in proportion to the system length $L$.
In particular, we compare the spin polarization on every site for the near-resonant eigenstate pair ($A,B$) and find the location ${x}\leq L$ where the last spin flip occurs. These two eigenstates are directly coupled with the $\hat{Z}_{{x}}$ operator ($|\bra{A}\hat{Z}_{{x}}\ket{B}|^2 \cong 4\epsilon^2$), and in the special case when ${x}=L$, the jump operator from the bath can directly couple the near-resonant pair (which we defined as the $Z$ case).

The normalized position of the last spin flip (${x}/L$) is illustrated in Fig.~\ref{fig:Zn} with different colors for different system sizes ($L\in[4,13]$). 
It is clear that as the system size gets larger, the curve peaks when it is smaller than 1 (the peak is near ${x}/L\approx 0.8 < 1$), indicating that an extensive number of sites are in between the bath and the last spin flip.

As we explained in the previous section, the minimal model including the possible near-resonance between ($A', B'$) predicts the matrix element proportional to $|\epsilon-\epsilon'|^2$.
The value of ${x}$ that yields the largest matrix element is determined by two competing factors: First, as $x$ becomes smaller, the resonance's range becomes shorter and fewer degrees of freedom are involved. Therefore the resonance becomes stronger, with a typically larger $\epsilon$.
However, decreasing $x$ by too much also causes the resonance to decouple from the bath: flipping the $\ell$-bit beside the bath may then produce a near copy of the resonance and result in small $|\epsilon - \epsilon'|$ even though $|\epsilon|$ and $|\epsilon'|$ are large, which will result in that resonance not being useful to the bath in relaxing the slowest mode. In the other direction, as $x$ becomes larger the resonance becomes longer-ranged and thus weaker ($|\epsilon|$ and $|\epsilon|'$ become smaller), but it also couples more strongly to the bath so that $|\epsilon - \epsilon'| \cong |\epsilon|\gg|\epsilon'|$.

Therefore, in most samples for large $L$ the bath uses near-resonances with $x/L<1$. The range of the optimal near-resonance is determined by a balance between the two considerations mentioned above.

\section{(vi) Conclusion: How the avalanche spreads near the avalanche instability in an infinite system}

Starting from a rare region characterized by weak randomness by chance, a local thermal bubble emerges and begins to thermalize adjacent typical MBL spins with a thermalization rate that exponentially decays ($\sim k^{-\ell}$) with the distance ($\ell$) from the rare region ($k>1$). 
As this thermal bubble expands and encompasses $2\ell$ typical spins (due to its spread in both directions along the spin chain), the many-body level spacing within the bubble becomes finer, reduced by a factor of $2^{-2\ell}$.  If this level spacing drops below the thermalization rate, the thermal bubble fails as a bath, and then the avalanche stops.  On the other hand, when the randomness is not sufficiently strong, the avalanche continues and gradually thermalizes the entire system. By comparing the many-body level spacing to the thermalization rate, one can determine the critical value of $k$ to be 4.

The aim of this letter is to investigate the propagation of avalanches near the critical point where $k=4$. We focus on the scenario where a large avalanche is spreading, and explore how this thermal bubble interacts with spins located a large distance ($d\gg1$ as depicted in Fig.~1 of the main text) from the rare region that initiated the avalanche.  In our model of this, we explicitly include the correct dynamics of spins 1 through $L$, with spin $L$ being closer to the rare region.  The next spin, which would be number $(L+1)$, we do not explicitly model.  This spin $(L+1)$ relaxes $4^L$ times faster than spin 1; spin 1 is the slowest mode that we are studying.  So on the time scales that we are studying, spin $(L+1)$ relaxes very rapidly, which makes it a reasonable approximation to model it and all the other thermalized spins between it and the rare region as a Markovian bath.  But spin $(L+1)$ is still very distant from the rare region, so it relaxes very slowly compared to the microscopic scales of our system.  We model this slowness of the bath due to spin $(L+1)$ by coupling the Markovian bath to spin $L$ with an arbtrarily small coupling $\gamma$.  Thus we have an almost closed system of $L$ spins, very weakly coupled to this bath.

%The thermalization timescale for spins at distance $d$ ($t>O(4^d)$) is much longer compared to the timescale for spins in closer proximity to the bath ($t<O(4^d)$). Consequently, the spins in between these two regions behave as an effective thermal bath. This motivates us to propose a model where the typical MBL spin chain is weakly connected to a structureless Markovian bath, serving as a replacement for spins beyond the length L.

By investigating this simplified avalanche model, we find a significant connection between many-body resonances and the dynamics of the spreading of the avalanche. Specifically, we observe that avalanches primarily propagate through the presence of strong rare near-resonances. The dominant processes that spread the avalanche involve only a few pairs of eigenstates of the closed system, including a strong near-resonance. To provide a detailed explanation of this mechanism, we present a minimal model of four eigenstates, incorporating two near-resonant eigenstates and two additional states, and numerically validated our proposed framework.

Our findings suggest a particular description of how avalanches propagate near the critical point: Consider a thermal bubble (an avalanche) that was initiated by a rare region and spread significantly. The spin at distance $\ell$ from the rare region relaxes at a rate $\sim k^{-\ell}$ with $k$ near 4, while the spins closer to the rare region relax faster than that.  The chain of spins between the rare region and spin $\ell$ has time to ``visit'' essentially all of the $2^{(\ell - 1)}$ eigenstates of that segment of the chain, and when it happens upon one of the special eigenstates, this allows the rare near-resonant process that flips spin $\ell$.  Then, as spin $\ell$ flips a few more times, the rest of that chain again visits essentially all of its eigenstates and happens upon one of the rare special states that allow the relaxation of spin $(\ell +1)$.  And this repeats for the next spin, etc., as long as the thermal bubble has a smaller many-body level spacing than the thermalization rate of the next relaxing spin, which will remain true if $k<4$.

%bubble by searching for specific eigenstates within the closed system of length $\ell$, which play a dominant role among the $2^\ell$ available states. Once these crucial eigenstates are identified, the system leverages rare near-resonances associated with them to initiate the thermalization process. Fully thermalizing the spin takes approximately $O(4^\ell)$ timescale. Subsequently, the avalanche proceeds to repeat this process as they again identify and utilize the new set of particular eigenstates to thermalize the spin at $\ell+1$.

\bibliography{supp}